\documentclass{article}
\usepackage{graphicx}
\usepackage{amsmath}

\newtheorem{theorem}{Theorem}

\newtheorem{definition}[theorem]{Definition}

\newtheorem{proposition}[theorem]{Proposition}

\newenvironment{proof}[1][Proof]{\textbf{#1.} }{\ \rule{0.5em}{0.5em}}

\date{}
\begin{document}
\date{}
\title{Wave propagation in periodic networks of thin fibers}
\author{S.
Molchanov\footnote{Dept of Mathematics, University of North
Carolina, Charlotte, NC 28223, smolchan@uncc.edu}, B.
Vainberg\footnote{Dept of Mathematics, University of North
Carolina, Charlotte, NC 28223, brvainbe@uncc.edu} \thanks{
The corresponding author}}\maketitle

\begin{abstract}
We will discuss a one-dimensional approximation for the problem of wave
propagation in networks of thin fibers. The main objective here is to describe
the boundary (gluing) conditions at branching points of the limiting
one-dimensional graph. The results will be applied to Mach-Zehnder
interferometers on chips and to periodic chains of the interferometers. The
latter allows us to find parameters which guarantee the transparency and
slowing down of wave packets.
\end{abstract}


{\it 2000 MCS:}
35J05; 35P25; 58J37; 78A40.

\section{Introduction.}

The paper concerns the asympototic theory of wave propagation in
networks $\Omega _{\varepsilon }$ of thin fibers when the thickness of
fibers $\varepsilon $ goes to zero. An approximation of a wave or heat
proceses in such a network by a one-dimensional problem on the limiting
metric (quantum) graph has been discussed in physical literature for at least
three decades. In recent years it was the subject of several
mathematical conferences. The central
point of the asymptotic theory is the structure of the physical field near
junctions (branching points of the network). In the majority of publications
on quantum graphs the gluing conditions on the vertices of the graph (they
correspond to the junctions) have the simplest Kirchhoff's form. This form
can be justified in some cases (say, for the heat transport in a network
with insulated walls, \cite{FW}, \cite{FW1}). We discussed some possible
applications of the quantum graph approximations to the study of periodic
optical systems: structure of the spectrum, scattering, slowing down of the
light, \cite{smv}-\cite{smv2}. We assumed there that the Kirchoff's GC at vertices were
imposed.

However, our recent study leads us to the conclusion that the Kirchoff's GC are an
exception in optical applications where the
spectral parameter $\lambda $ is greater than the threshold $\lambda _{0}.$
The latter is equivalent to the condition that the propagation of waves in
the waveguides (cylindrical parts of the network) is possible. If $\Omega
_{\varepsilon }$ is unbounded, this also means that $\lambda $ belongs to
the absolutely continuous spectrum of the problem. While many particular
cases of that problem with $\lambda =\lambda _{0}+O(\varepsilon ^{2})$ or $
\lambda <\lambda _{0}$ were considered (see \cite{IT}-\cite{RS}), the
publications \cite{MV}-\cite{MVspain} were the first ones dealing with the case $
\lambda \geq \lambda _{0},$ and the first ones where the significance of the
scattering solutions for asymptotic analysis was
established. Papers \cite{MV}-\cite{MVspain} contain asymptotic
analysis of the spectrum, resolvent and solutions of the problem in a
network $\Omega _{\varepsilon }$ when $\varepsilon \rightarrow 0$. It was shown there that
the GC in those cases have general symplectic structure and can be expressed in
terms of the scattering matrices defined by individual junctions.

The main goal of the present paper is to describe the wave propagation through networks
of thin fibers of necklace type resulting in general GC at the vertices of the limiting graph. We say that a network is of necklace type if it is periodic in one direction and is bounded in the orthogonal plane.
The transition from the networks
of thin fibers to the one-dimensional problem on the graph will be recalled in the next section (see details in \cite{MV}-\cite{MVspain}). The necklace type graphs will be considered in section 3. We will calculate the propagator through one period, find the dispersion relation, describe the band-gap structure of the spectrum and find reflection and transmission coefficients for the truncated graph.

The main feature which distinguishes the graph theory from the Bloch theory of 1D periodic Schrodinger (Hill) operators is that the
propagator through one period is not an analytic function of the frequency anymore, but a meromorphic one. The corresponding poles (resonances) play an important role in applications. In the last section, we consider a specific necklace device and show how the earlier results allow one to find parameters which provide slowing down of the wave packets (slowing down of the light) accompanied by the transparency (almost zero reflection).

\section{Transition from networks to quantum graphs.}

Consider the stationary wave (Helmholtz) equation
\begin{equation}
H_{\varepsilon }u=-\Delta u=\omega^2 u,\text{ \ \ \ }x\in
\Omega _{\varepsilon },\text{ \ \ \ }Bu=0\text{ \ \ on }\partial \Omega
_{\varepsilon }\text{,}  \label{h0}
\end{equation}
in a domain $\Omega _{\varepsilon }\subset R^{d},$ $d\geq 2,$ with
infinitely smooth boundary (for simplicity), which has the following
structure: $\Omega _{\varepsilon }$ is a union of a finite number of
cylinders $C_{j,\varepsilon }$ (which will be called channels) of lengths $%
l_{j},$ $1\leq j\leq N,$ with diameters of cross-sections of order $O\left(
\varepsilon \right) $ and domains
(which will be called junctions) connecting the channels
into a network. It is assumed that the junctions have diameters of the same
order $O(\varepsilon )$. The boundary condition has the form: $B=1$ (the
Dirichlet BC) or $B=\frac{\partial }{\partial n}$ (the Neumann BC) or $%
B= \frac{\partial }{\partial n}+\alpha (x),$ where $n$ is the
exterior normal and\ the function $\alpha \geq 0$ is real valued and does
not depend on the longitudinal (parallel to the axis) coordinate on the
boundary of the channels. One also can impose one type of BC on the lateral
boundary of $\Omega _{\varepsilon }$ and another BC on the free ends (which are
not adjacent to a junction) of the channels.

\begin{figure}[tbph]
\begin{center}
\includegraphics[width=0.7\columnwidth]{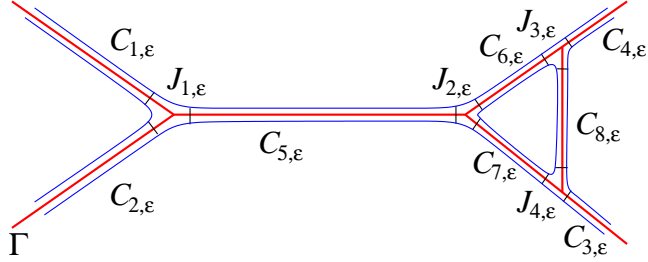}
\end{center}
\caption{An example of a domain $\Omega _{\protect\varepsilon }$ with four
junctions, four unbounded channels and four bounded channels.}
\label{fig-1}
\end{figure}

The domain $\Omega _{\varepsilon }$ shrinks to a one-dimensional metric graph $\Gamma$ as
$\varepsilon \rightarrow 0$. The axes of the
channels form edges $\Gamma _{j},$ $1\leq j\leq N,$ of $\Gamma $, and the distances between points of
$\Gamma _{j}$ are defined by the distances between the corresponding points of
the channels. The junctions
shrink to vertices of the graph $\Gamma $.
We denote the set of vertices $v_{j}$ by $V$.

\begin{figure}[tbph]
\begin{center}
\includegraphics[width=0.8\columnwidth]{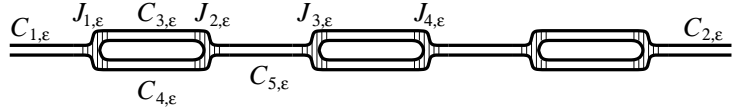}
\end{center}
\caption{Another example of a domain $\Omega _{\protect\varepsilon }$, network of a necklace type which will be studied in more detail later.}
\label{fig-2}
\end{figure}

For the sake of simplicity, we impose the following two geometrical conditions. First, we
assume that all the channels $C_{j,\varepsilon }$ have the same
cross-section $\pi _{\varepsilon } $ (the general case is studied in \cite{MV1}
). The second condition concerns the junctions.
We assume that they are self-similar. The latter means that
there exist an $\varepsilon$- independent domain $J_{v}$ and a point $\widehat{x}=\widehat{x}_v$
such that
\begin{equation}
J_{v,\varepsilon }=\{(\widehat{x}+\varepsilon x):x\in J_{v}\}.
\label{hom}
\end{equation}

From the self-similarity assumption it follows that $\pi _{\varepsilon
}\ $is an $\varepsilon $-homothety of a bounded domain $\pi \subset
R^{d-1}$.

Let $\lambda _{0}<\lambda _{1}\leq \lambda _{2}...$ be the eigenvalues of the
negative Laplacian $-\Delta _{d-1}$ in $\pi $ with orthonormal eigenfunctions $\{\varphi _{n}(y)\}$,
\[
-\Delta _{d-1}\varphi _{n}(y)=\lambda _{n}\varphi _{n}(y),~~Bu=0~~~\text {on}~~\partial \pi ,
\]
where $B$ is the boundary operator on
the channels defined in (\ref{h0}).
 Then $\varepsilon^{-2}\lambda _{n}$
are eigenvalues of $-\Delta _{d-1}$ in $\pi _{\varepsilon
}$ and $\{\varepsilon ^{-d/2}\varphi _{n}(y/\varepsilon )\}$ are the
corresponding orthonormal eigenfunctions,
\[
-\Delta _{d-1}\varphi _{n}(y/\varepsilon)=\varepsilon^{-2}\lambda _{n}\varphi _{n}(y/\varepsilon),~~Bu=0~~~\text {on}~~\partial \pi_{\varepsilon.}
\]
We will call the point $\varepsilon^{-2}\lambda _{0}$ the
threshold, since it is the bottom of the absolutely continuous spectrum
of operator (\ref{h0}) if $\Omega_{\varepsilon}$ has an infinite channel.

We introduce Euclidean coordinates $(z,y)$ in channels $C_{j,\varepsilon }$
chosen in such a way that the $z$-axis is parallel to the axis of the
\begin{figure}[tbph]
\begin{center}
\includegraphics[width=0.4\columnwidth]{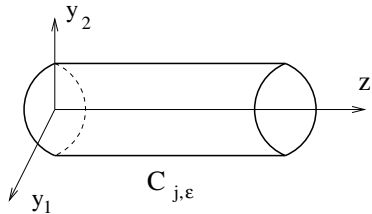}
\end{center}
\caption{Local coordinates in a channel.}
\label{fig-4}
\end{figure}
channel, hyperplane $R_{y}^{d-1}$
is orthogonal to the axis, and $C_{j,\varepsilon }$ has the following form
in the new coordinates:
\begin{equation*}
C_{j,\varepsilon }=\{(z,\varepsilon y):0<z<l_{j},\text{ }y\in \pi \}.
\end{equation*}
If a channel $C_{j,\varepsilon }$ is bounded ($l_{j}<\infty $), the
direction of the $z$ axis can be chosen arbitrarily (at least for now). If a channel is unbounded, then $z=0$ corresponds to its
cross-section which is adjacent to the junction.

We will impose the condition $(\varepsilon\omega)^2 \in (\lambda_0,\lambda_1)$ (see \cite{MV1} for the
general case). Note that
waves governed by (\ref{h0}) do not propagate through the channels if $(\varepsilon\omega)^2 <\lambda_0$.
There exists only one propagating mode
\begin{equation} \label{disp}
e^{\pm i\sigma z}\varphi
_{0}(y/\varepsilon ),~~\sigma=\sqrt{\omega ^2 -\varepsilon^{-2}\lambda _{0}},
\end{equation}
 if
$(\varepsilon\omega)^2 \in (\lambda_0,\lambda_1)$, and there are many similar modes
 \[
 e^{\pm i\sigma _j z}\varphi _{j}(y/\varepsilon ),~~\sigma_j=\sqrt{\omega ^2 -\varepsilon^{-2}\lambda _{j}},
 \]
 if $(\varepsilon\omega)^2 > \lambda_j$.

While many particular
cases of problem (\ref{h0}) with $(\varepsilon\omega)^2 =\lambda _{0}+O(\varepsilon ^{2})$ or $
(\varepsilon\omega)^2 <\lambda _{0}$ were considered (see references in \cite{MVspain}), the
publications \cite{MV}-\cite{MVspain} were the first ones dealing with the case $
(\varepsilon\omega)^2 \geq \lambda _{0}, \varepsilon \rightarrow 0,$ and the first ones where the significance of the
scattering solutions for asymptotic analysis of $H_{\varepsilon }$ was
established. In particular, it was shown there that in both cases $(\varepsilon\omega)^2
>\lambda _{0}$ and $(\varepsilon\omega)^2\approx \lambda _{0},$ the scattering solutions for equation (\ref{h0})
and the resolvent of the operator  $H_{\varepsilon }$ can be approximated by the corresponding solutions of the one-dimensional problem on the limiting graph
$\Gamma $ with the GC expressed in terms of the
scattering matrices of the individual extended junctions.

Let us recall the definition of scattering solutions for the problem (\ref
{h0}) in $\Omega _{\varepsilon }$ when $(\varepsilon\omega)^2 \in (\lambda _{0},\lambda
_{1}).$ The scattering solution $\Psi =\Psi
_{p,\varepsilon }$ describes the propagation of an incident wave with unit amplitude and frequency $\omega$ coming through the channel $C_{p,\varepsilon}$.
\begin{definition}
\label{d2}Let $\lambda _{0}<(\varepsilon\omega)^2 <\lambda _{1}.$ A function $\Psi =\Psi
_{p,\varepsilon },$ $1\leq p\leq m,$ is called a solution of the
scattering problem in $\Omega _{\varepsilon }$ if
\begin{equation}
(-\Delta -\omega ^2 )\Psi =0,\text{ \ }x\in \Omega
_{\varepsilon };\text{ \ \ \ }B\Psi =0\text{ \ on }\partial \Omega
_{\varepsilon },  \label{b9}
\end{equation}
and $\Psi $ has the following asymptotic behavior in infinite channels $%
C_{j,\varepsilon },\ \ 1\leq j\leq m:$%
\begin{equation}
\Psi _{p,\varepsilon }=[\delta _{p,j}e^{-i\sigma z}+t_{p,j}e^{i\sigma z}]\varphi _{0}(y/\varepsilon )+O(e^{-\frac{\alpha z}{%
\varepsilon }}),~~z\rightarrow \infty ,\text{ \ }\alpha >0.  \label{b10}
\end{equation}
Here $\sigma=\sqrt{\omega ^2 -\varepsilon^{-2}\lambda _{0}}$, $\delta _{p,j}$ is the Kronecker symbol, i.e. $\delta _{p,j}=1$ if $%
p=j, $ $\delta _{p,j}=0$ if $p$ $\neq j.$
\end{definition}
\textbf{Remark.} The term with
the coefficient $\delta _{p,j}$ in (\ref{b10}) corresponds to the incident
wave (coming through the channel $C_{p,\varepsilon }$), $t_{p,p}$ is the reflection coefficient, the terms with
coefficients $t_{p,j}, j \neq p,$ describe the transmitted waves. The
coefficients $t_{p,j}=t_{p,j}(\varepsilon ,\omega )$ depend on $\varepsilon
$ and $\omega $. The matrix

\begin{equation}
T=[t_{p,j}]  \label{scm}
\end{equation}
is called\textit{\ the scattering matrix}.

Standard arguments based on the Green formula provide the following
statement.

\begin{theorem}
\label{t3}When $\lambda _{0}<(\varepsilon\omega)^2 <\lambda _{1},$\ the scattering matrix $%
T$ is unitary and symmetric ($t_{p,j}=t_{j,p}$).
\end{theorem}

It happens that the scattering solutions $\Psi _{p,\varepsilon}$ can be approximated with
an exponential in $\varepsilon$ accuracy using the scattering solutions of a one-dimensional problem on the limiting graph $\Gamma$ which are defined as follows. Consider the following equation on $\Gamma$
\begin{equation}
-\frac{d^{2}}{dz^{2}}\psi=\sigma ^2\psi ,~~\sigma=\sqrt{\omega ^2 -\varepsilon^{-2}\lambda _{0}}.
\label{greq}
\end{equation}
Obviously,
\[
\psi=c_{1,j}e^{-i\sigma z}+c_{2,j}e^{i\sigma z}
\]
on the edges $\Gamma_j \subset \Gamma$.

 We split the set $V$ of vertices $v$
of the graph into two subsets $V=V_{1}\cup V_{2},$ where the vertices from
the set $V_{1}$ have degree $1$ and correspond to the free ends of the
channels, and the vertices from the set $V_{2}$ have degree at least two and
correspond to junctions.

\begin{definition}We will say that $%
\psi =\psi _{p,\varepsilon}$ is a solution of the scattering problem on the
graph
$\Gamma $ with the incident wave coming through the edge $\Gamma _{p}$
if $\psi _{p,\varepsilon}$ satisfies equation (\ref{greq}),
\begin{equation}
\psi _{p,\varepsilon}(\gamma)=\delta _{p,j}e^{-i\sigma z}
+t _{p,j}e^{i\sigma z},~\gamma \in \Gamma _j,
\label{2222}
\end{equation}
on infinite edges $\Gamma _j, 1\leq j \leq m$, and satisfies the following GC at vertices $v$ of $\Gamma$:
\begin{equation}
B\psi =0\text{ \ \ at }v\in V_{1},  \label{bce}
\end{equation}
\begin{equation}
i \lbrack I_{v}+T_{v}(\varepsilon \omega )]\frac{d}{dz}\psi ^{(v)}(z)-%
\sigma[I_{v}-T_{v}(\varepsilon \omega )]\psi ^{(v)}(z)=0,%
\text{ \ \ \ }z=0,\text{ \ \ \ }v\in V_{2}.  \label{gc}
\end{equation}
\end{definition}
We keep the same BC at $%
v\in V_{1}$ as at the free end of the corresponding channel of $\Omega
_{\varepsilon }$, see (\ref{h0}), and we will specify GC
(\ref{gc}) in the next two paragraphs. However, first we would like to stress
that scattering coefficients $t_{p,j}$ in (\ref{2222}) are not required to coincide with those defined in (\ref{b10}) for the problem in the domain $\Omega_{\varepsilon}$. With the appropriate choice of GC (\ref{gc}), these
 coefficients are the same, and therefore we use the same notation.

We choose the parametrization on $\Gamma $ in such a way that $z=0$
at $v$ for all edges adjacent to this particular vertex. Let $d=d(v)\geq 2$
be the order (the number of adjacent edges) of the vertex $v\in V_{2}.$
For any function $\psi $ on $\Gamma ,$ we form a column-vector $%
\psi ^{(v)}=\psi ^{(v)}(z)$ with $d(v)$ components which is
formed by the restrictions of $\psi $ on the edges of $\Gamma $
adjacent to $v.$ We will need this vector only for small values of $z\geq
0. $ The GC (\ref{gc}) are defined in terms of
auxiliary scattering problems for extended junctions $J _{v,\varepsilon}
^{\infty}$. Each extended junction
consists of $J_{v,\varepsilon }$ and all the channels adjacent to $J_{v,\varepsilon}$.
If some of these channels have finite length, we extend them to infinity (see Fig. \ref{extj}).
\begin{figure}[tbph]
\begin{center}
\includegraphics[width=0.8\columnwidth]{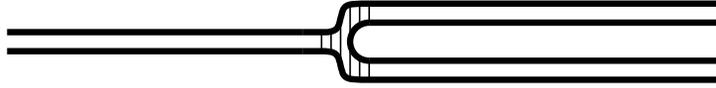}
\end{center}
\caption{An extended junction for a necklace waveguides.}
\label{extj}
\end{figure}
The matrix $T=T_{v}(\varepsilon\omega )$ is the scattering matrix
for the problem (\ref{h0}) in $J _{v,\varepsilon
}^{\infty }$ and $I_{v}$ is the unit matrix of the same size as the size
of $T.$ Note that the self-similarity of $J _{v,\varepsilon
}^{\infty }$ implies that $T=T_{v}(\varepsilon\omega )$ depends only on the product $\varepsilon\omega$. Hence, $T_v$ can be determined by solving the scattering problem in the corresponding extended junction with
$\varepsilon=1.$ Now (\ref{gc}) is defined. We need only to take components of the vector $\psi ^{(v)}$ in the same
order as the order of channels of $J _{v}^{\infty }.$

Denote by $F(\varepsilon)$ the set $\{\omega_j\}$ of the values of $\omega$ for which problem (\ref{greq}), (\ref{bce}), (\ref{gc}) has a nontrivial solution $\psi \in L^2(\Gamma)$ and
 $(\varepsilon\omega_j)^2\in[\lambda_0,\lambda_1]$. Note that both equation and the boundary conditions of the problem depend on $\omega$ and $\varepsilon$. Let $F^0$ be the set $\{\omega_j\}$ of values of $\omega$ such that $\omega_j^2$ is an eigenvalue of the operator $H_{\varepsilon}$ (see (\ref{h0})) in one of the domains
$J _{v,\varepsilon}^{\infty }$ and $(\varepsilon\omega_j)^2\in [\lambda_0,\lambda_1]$.
Due to self-similarity of the extended junctions
$J _{v,\varepsilon}^{\infty }$, the set $\varepsilon F^0$ does not depend on $\varepsilon$. Let $F^{\nu}$ be the
$e^{\frac{%
-\nu}{\varepsilon }}$-neighborhood of the set $F ^{0}\cup
F (\varepsilon ).$
\begin{theorem}(\cite{MVspain})\label{t1}
1) The set $F ^{0}\cup
F (\varepsilon )$ has finitely many points (the number of points depends on $\varepsilon$).

2) For any interval $[\lambda_0 ,\lambda ^{\prime }],~ \lambda ^{\prime }< \lambda_1,$ there exist $\rho ,\nu >0$
such that scattering solutions $\Psi _{p,\varepsilon}(x)$ of the problem in
$\Omega_ \varepsilon $ have the following asymptotic behavior on the channels
of $\Omega_ \varepsilon $ as $\varepsilon \rightarrow 0$%
\begin{equation*}
\Psi _{p,\varepsilon}(x)=\psi _{p}^{(\varepsilon)}(\gamma )\varphi _{0}(\frac{y}{\varepsilon }%
)+r_{p}^{(\varepsilon )}(x),~~\gamma \in \Gamma,
\end{equation*}
where $\psi _{p}^{(\varepsilon )}(\gamma )$ are the
scattering solutions  of the problem on the graph $\Gamma $ and
\begin{equation*}
|r_{p}^{(\varepsilon )}(x)|\leq Ce^{\frac{-\rho d(\gamma )}{\varepsilon }},%
\text{ \ \ \ \ \ \ }(\varepsilon\omega)^2 \in [\lambda_0 ,\lambda ^{\prime
}], ~\omega \notin F ^{\nu }.
\end{equation*}
Here $\gamma =\gamma (x)$ is the point on $\Gamma $ which is defined by the
cross-section of the channel through the point $x,$ and $d(\gamma )$ is the
distance between $\gamma $ and the closest vertex of the graph.
\end{theorem}

Note that this theorem implies the coincidence of the scattering matrices of the problems on
$\Omega_ \varepsilon $ and on the graph $\Gamma$.

We will conclude this section by the following important proposition:
\begin{proposition}
Suppose that det$[I_{v}+T_{v}(\varepsilon \omega )]$ is not identically equal to zero. Then for all $\omega$ such that
 $~(\varepsilon\omega)^2 \in [\lambda_0, \lambda_1]$, except at most a finite number of points, the GC
(\ref{gc}) can be written in the form
\begin{equation}
\frac{d}{dt}\psi ^{(v)}(t)-
\sigma A_v\psi ^{(v)}(t)=0,~~t=0,~~\sigma^2=\sqrt{\omega ^2 -\varepsilon^{-2}\lambda _{0}},
~~v\in V_{2}, \label{gc1}
\end{equation}
where matrix
\[
A_v=-i[I_{v}+T_{v}(\varepsilon\omega )]^{-1}
[I_{v}-T_{v}(\varepsilon\omega )]
\]
 is real valued and symmetric ($A_v=A'_v$).
\end{proposition}
\begin{proof}Due to analyticity of $T_v$ in $\lambda$, we need only to justify properties of $A_v$. From
 Theorem \ref{t3} it follows that the eigenvectors of $T_v$ can be chosen to be real valued, i.e., there
 exists a real-valued orthogonal matrix $C_v=C_v(\lambda)$ such that $T_v=C_v^{-1}D_vC_v$ where $D_v$
 is a diagonal matrix with elements $z_j,~|z_j|=1$ on the diagonal. Numbers $z_j$ are (complex) eigenvalues of $T_v$. Obviously,
 \[
I_{v}\pm T_{v}=I_{v}\pm C_v^{-1}D_vC_v=C_v^{-1}[I_{v}\pm D_v]C_v,~~
A_v=iC_v^{-1}[I_{v}+ D_v]^{-1}[I_{v}- D_v]C_v.
 \]
 It remains to note that the numbers $i\frac{1-z_j}{1+z_j}$ are real.
\end{proof}
\section{Necklace graphs}
 In this section we will consider a periodic metric graph $
\Gamma $ of the necklace type (see \cite{smv1}). It has the following form: one cell of periodicity consists of two
arches $\mathcal{L}_{1},\mathcal{L}_{2}$ of lengths $l_{1}$ and $l_{2}\leq
l_{1}$ connected at end points and of a segment $\mathcal{L}_{3}$ of length $%
l_{3}$ starting at one of these points (see fig. 1). We assume that the
necklace is placed horizontally. Let $v_{2m-1},v_{2m}$ be the left and the
right end points of the segments with $v_{0}$ being the origin, and with two
arches connecting the points $v_{2m}$ and $v_{2m+1}$, ~~ $m=0,\pm 1,\pm
2,...~.$

We introduce two related local coordinates $z$ and $s$ on the edges of the graph. Both are
the lengths of the corresponding part of an arch or a segment, measured from
some end of the edge. When a
neighborhood of some vertex is considered (for example when GC are defined), the distance is measured from that vertex for all the edges adjacent to this vertex. The coordinate (distance) $z$
is used in this case. In other cases it will be convenient for us to measure the distance from the left end
of the edge to the right. We will specify this situation by using parameter $s$ instead of $z$.  Thus, $s=z$ or $s=l-z$ where $l$ is the length of the edge.

We equip the graph with the natural Lebesgue measure and
consider the Hamiltonian $H$ on $L^{2}(\Gamma )$ given by $H=-\frac{d^{2}}{dz^{2}}=-\frac{d^{2}}{ds^{2}}$ on the graph with the following GC
(see (\ref{gc1}))
\begin{equation}
\frac{d}{dt}\psi ^{(v)}(z)-
\sigma A_v(\varepsilon\omega)\psi ^{(v)}(z)=0,%
\text{ \ \ \ }z=0,\text{ \ \ \ }v\in V, \label{gc2}
\end{equation}
at the vertices of $\Gamma$. Here and below we use
\[
\sigma =\sqrt{\omega^2-\varepsilon^{-2}\lambda_0}.
\]

Our goal in this section is to study the propagation of waves on $\Gamma$ governed by the equation
\begin{equation} \label {eq}
H\psi=-\frac{d^{2}}{dz^{2}}\psi =\sigma ^{2}\psi,~~\gamma \in \Gamma,
\end{equation}
and GC (\ref{gc2}). Note that both the equation and GC depend on the frequency $\omega$ and $\varepsilon$.

We enumerate the components $\psi_j$ of the vector $\psi^{(v)}$ in the following order: $\psi_j$ corresponds to the edge of the length $l_j$. Thus the first two components of the vector correspond to the shoulders of the loop, and the third component corresponds to the straight edge.

\begin{figure}[tbph]
\begin{center}
\includegraphics[width=0.8\columnwidth]{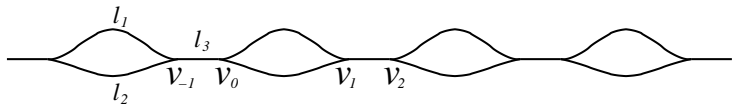}
\end{center}
\caption{Graph $\Gamma$ for the necklace waveguide.}
\label{neckgr}
\end{figure}

The goal of this section is to
define and evaluate the (Pr\"{u}ffer) monodromy operator $M_{\sigma}$ (transfer operator over the period) for the problem (\ref{eq}), (\ref{gc2}).
Let $\psi'=\frac{d\psi}{ds}$ and $\left(
\begin{array}{c}
\psi \\
\sigma ^{-1}\psi'
\end{array}
\right) (\alpha)$ be the Cauchy data (it always will have the factor $\sigma
^{-1} $ in the second component) of the solution $\psi$ of the equation
(\ref{eq}) evaluated at the point $\alpha$ of a straight segment of $\Gamma $.
When $\alpha=v_{n}$ is one of the end points of a segment, we understand
this vector as the limit of the corresponding vectors evaluated at $\alpha$ as $\alpha $
approaches $v_{n}$ moving along the segment (not along one of the arches).
We denote by $M_{\sigma }$ the monodromy operator:
\begin{equation*}
M_{\sigma }:\left(
\begin{array}{c}
\psi \\
\sigma ^{-1}\psi'
\end{array}
\right) (v_{0})\rightarrow \left(
\begin{array}{c}
\psi \\
\sigma ^{-1}\psi '
\end{array}
\right) (v_{2})
\end{equation*}
and we denote by $T_{\sigma }$ the (Pr\"{u}ffer) transfer
operator over the loop:
\begin{equation*}
T_{\sigma }:\left(
\begin{array}{c}
\psi \\
\sigma ^{-1}\psi'
\end{array}
\right) (v_{0})\rightarrow \left(
\begin{array}{c}
\psi \\
\sigma ^{-1}\psi'
\end{array}
\right) (v_{1}).
\end{equation*}
We will use the same notations $M_{\sigma }, T_{\sigma}$ for the matrices of the operators as for the corresponding operators.

Let us write matrix $A_v=(a_{i,j})$ (see (\ref{gc1}), (\ref{gc2})) in the form
\begin{equation} \label{AB}
A_v=\left(
\begin{array}{c}
B~\delta \\
\delta^*~c
\end{array}
\right),~~\text{where}~
B=\left(
\begin{array}{c}
a_{1,1} ~a_{1,2}\\
a_{2,1}~a_{2,2}
\end{array}
\right),~\delta=
\left(
\begin{array}{c}
\delta_1 \\
\delta_2
\end{array}
\right)=\left(
\begin{array}{c}
a_{1,3} \\
a_{2,3}
\end{array}
\right), ~c=a_{3,3}.
\end{equation}
We will need the following matrices
\begin{equation} \label{csp}
S=\left(\begin{array}{c}
\sin\sigma l_1~~~~0 \\
0~~~~\sin\sigma l_2
\end{array}\right),~~C=\left(\begin{array}{c}
\cos\sigma l_1~~~~0 \\
0~~~~\cos\sigma l_2
\end{array}\right),  ~~P=C+SB,
\end{equation}
\begin{equation} \label{mn}
 M=(I-P^2)^{-1}PS,~~N=(I-P^2)^{-1}S.
\end{equation}
\begin{theorem}
1) Matrix $T_\sigma$ has the form
\begin{equation} \label{trl}
T_{\sigma }=-\left(
\begin{array}{c}
\frac{m}{n}~~~\frac{1}{n} \\
\frac{m^2-n^2}{n}~~\frac{m}{n}
\end{array}\right),~~~n=\langle\delta, N\delta\rangle,~~m=c+\langle\delta, M\delta\rangle.
\end{equation}
2) Matrix $M_{\sigma}$ has the form
\begin{equation} \label{mon}
M_{\sigma }=\left(\begin{array}{c}
\cos\sigma l_3~~~\sin\sigma l_3 \\
-\sin\sigma l_3~~~~\cos\sigma l_3
\end{array}\right)T_{\sigma }.
\end{equation}
\end{theorem}
\begin{proof}Let $\psi_j= \psi_j(s),~ j=1,2,$ be the restrictions of the function $\psi$ on the upper and lower arches of the graph between the points $v_0$ and $v_1$. Obviously,
\[
\psi_1(s)=\frac{\psi_1(0)\sin\sigma(l_1-s)+\psi_1(l_1)\sin\sigma s}{\sin\sigma l_1}.
\]
Thus,
\[
\psi_1'(0)=\frac{-\sigma\psi_1(0)\cos \sigma l_1+\sigma \psi_1(l_1)}{\sin\sigma l_1}, ~~
\psi_1'(l_1)= \frac{-\sigma \psi_1(0)+\sigma\psi_1(l_1)\cos \sigma l_1}{\sin\sigma l_1}.
\]
A similar formula is valid for $\psi_2$. Hence, taking into account the relation between parameters $s$ and $t$, we obtain the following connection between the values of the vector $\phi=\left(\begin{array}{c}
\psi_1 \\\psi_2 \end{array}\right)$ and its derivative $\frac{d\phi}{dz}$ at points $v_0$ and $v_1$
of the graph:
\begin{equation*}
\frac{d\phi}{dz}(v_0)=-\sigma CS^{-1}\phi(v_0)+\sigma S^{-1}\phi(v_1),~~
\frac{d\phi}{dz}(v_1)=-\sigma CS^{-1}\phi(v_1)+\sigma S^{-1}\phi(v_0).
\end{equation*}

Let us denote the restrictions of $\psi$ on the straight edges of the graph by $\psi_3.$ Consider vectors
$\psi^{(v_0)}=\left(\begin{array}{c}
\phi(v_0) \\ \psi_3(v_0) \end{array}\right),~\psi^{(v_1)}=\left(\begin{array}{c}
\phi(v_1) \\ \psi_3(v_1) \end{array}\right).$ GC (\ref{gc2}) implies
\begin{equation*}
\left(\begin{array}{c}
-CS^{-1}\phi(v_0)+ S^{-1}\phi(v_1)\\ \omega ^{-1}\frac{d}{dz}\psi_3(v_0) \end{array}\right)=
\left(
\begin{array}{c}
B~~~\delta\\
\delta ^*~~c
\end{array}
\right)\left(\begin{array}{c}
\phi(v_0) \\ \psi_3(v_0) \end{array}\right),
\end{equation*}
\begin{equation*}
\left(\begin{array}{c}
-CS^{-1}\phi(v_1)+ S^{-1}\phi(v_0)\\ \sigma ^{-1}\frac{d}{dz}\psi_3(v_1) \end{array}\right)=
\left(
\begin{array}{c}
B~~~\delta \\
\delta^*~~c
\end{array}
\right)\left(\begin{array}{c}
\phi(v_1) \\ \psi_3(v_1) \end{array}\right).
\end{equation*}
These equations can be rewritten in the following form
\[
-(CS^{-1}+B)\phi(v_0)+ S^{-1}\phi(v_1)=\delta\psi_3(v_0),
\]
\[
-\sigma ^{-1}\frac{d}{ds}\psi_3(v_0)=\delta^*\phi(v_0)+c\psi_3(v_0),
\]
\[
-(CS^{-1}+B)\phi(v_1)+ S^{-1}\phi(v_0)=\delta\psi_3(v_1),
\]
\[
\sigma ^{-1}\frac{d}{ds}\psi_3(v_1)=\delta^*\phi(v_1)+c\psi_3(v_1).
\]
We multiply the first and third equations by $S$, replace $C+SB$ by $P$ (see (\ref{csp})) and then solve these equations for $\phi(v_0),~\phi(v_1)$. This implies
\[
\phi(v_0)=-(I-P^2)^{-1}PS\delta \psi_3(v_0)+(I-P^2)^{-1}S\delta \psi_3(v_1),
\]
\[
\phi(v_1)=-(I-P^2)^{-1}S\delta \psi_3(v_0)+(I-P^2)^{-1}PS\delta \psi_3(v_1).
\]
We substitute these relations into the second and forth equations of the system above and then solve those equations for $\psi_3(v_1), ~\sigma ^{-1}\frac{d}{ds}\psi_3(v_1)$. This provides the transfer operator over the loop with the transfer matrix defined in (\ref{trl}). This completes the proof of the first part of theorem. The second statement of the theorem is obvious, since the left factor in the right-hand side of (\ref{mon}) is the transfer matrix over the segment $[v_0,v_1]$ of the graph.
\end{proof}
\section{Some applications of necklace waveguides}
We will discuss here two practical features of necklace waveguides. The first concerns slowing down of the light ( slowing down of propagation of wave packets) in these waveguides. There is an extended literature on the principles of this phenomenon, possible applications and practical devices. Usually some type of a periodic structure is suggested for these devices with a band-gap structure of the spectrum. If a narrow band is created, then the dispersion relation for the corresponding frequency is flat and the group velocity $V_g$ is small.

The main feature which distinguishes a necklace waveguide from other one-dimensional (or quasi one-dimensional) periodic problem is the following. While the Hill discriminant for a standard periodic Schrodinger operator is an analytic function of frequency, it is meromorphic for the necklace waveguides. We will show that one can easily find parameters when the Hill discriminant has two close poles (as close as one pleases) with a zero in between. Thus the band will be as narrow as we please around a chosen value $\sigma=\sigma_0$ of the frequency. Therefore, the group velocity will be small if the support of the wave packets belongs to a small neighborhood of $\sigma_0$.
\begin{figure}[tbp]
\par
\begin{center}
\includegraphics[width=110mm, angle=0]{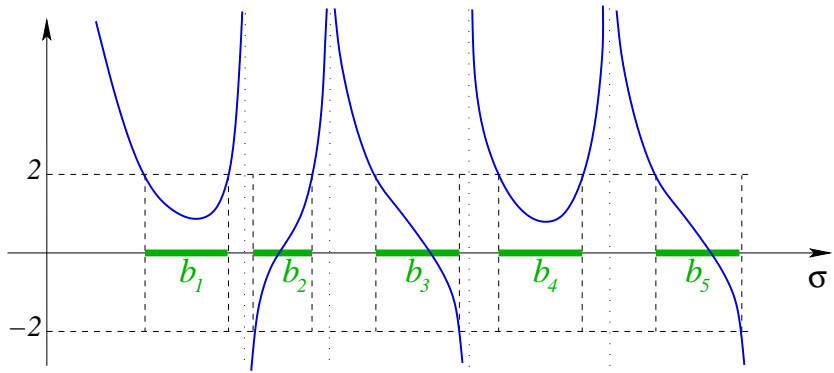}
\caption{The graph of the Hill discriminant $\mathrm{Trace}M_{\protect\sigma }=2\cos
k(\sigma ),~~\sigma=\sqrt{\omega^2-\varepsilon^{-2}\lambda_0}.$}
\label{neck22}
\end{center}
\end{figure}

The second feature concerns the truncated necklace graph (waveguide) $\Gamma_N$ which consists of $N$ cells of periodicity confined between points $v_0$ and $v_{2N}$ and the rays $(-\infty,v_0), (v_{2N},\infty.)$ When propagation of a narrow in frequency wave packet through a finite device is considered, it could happen that the waves, which are slowing down in the device, can not enter the device. In other words, one needs to know that the reflection coefficient $r$ for the truncated necklace waveguide is not too big (better if it is small) for the frequencies where the slowing down occurs. It will be shown that for a given frequency $\sigma=\sigma_0$ one can choose the parameters $l_j, j=1,2,3,$ such that $V_g=r=0$ at $\sigma=\sigma_0$, and therefore they are small for $\sigma$ close enough to $\sigma_0$.

\textit{Slowing down, preliminary discussion}. Theorem \ref{t1} allows us to reduce the study of propagation of single frequency waves and wave packets in periodic necklace waveguide $\Omega_{\varepsilon}$ (see Fig. \ref{fig-2}) to a study of the corresponding problem on the necklace graph $\Gamma$ (see Fig. \ref{neckgr}). One needs only to take in correspondence the frequencies $\omega$ and $\sigma$ of the waves in $\Omega_{\varepsilon}$ and on $\Gamma$ using the relation
\[
\sigma=\sqrt{\omega^2-\varepsilon^{-2}\lambda_0}.
\]
%
%

The spectrum of periodic problem (\ref{eq}), (\ref{gc2}) has a band-gap structure with the bands on the $\sigma$-axis defined by the inequality
\[
|F(\sigma)|<2,~~~F(\sigma)=\text{Tr}M_{\sigma},
\]
where function $F$ (called the Hill discriminant) is equal to the trace of the monodromy matrix (see Fig \ref{neck22}). The same function defines the dispersion relation (we will write it in the form $k=k(\sigma)$) of the problem on $\Gamma$:
\begin{equation} \label{dispr}
\cos k(\sigma)=\frac{1}{2}F(\sigma)=\frac{1}{2}\text{Tr}M_{\sigma}.
\end{equation}
Let $\omega=\omega_0$ belong to the frequency interval of a narrow wave packet of the problem in $\Omega_{\varepsilon}$ and
\[
\sigma_0=\sqrt{\omega_0^2-\varepsilon^{-2}\lambda_0}.
\]
We will find parameters $l_j$ in such a way that
\begin{equation} \label{f0}
|F(\sigma_0)|=0,
\end{equation}
 and there is a point $\sigma=\sigma_1$, for example, to the right of $\sigma_0$, such that $F(\sigma)$ has a pole at $\sigma=\sigma_1$ and  $|\sigma_0-\sigma_1| = \epsilon \ll 1$. Then there is a point $\sigma' \in (\sigma_0,\sigma_1)$ such that
 \[
 |F(\sigma)|<1 ~~\text{for}~~ \sigma \in \Delta=(\sigma_0,\sigma'),~~|F(\sigma')|=1.
 \]
 Then interval $\Delta$ belongs to a band. From (\ref{dispr}) it follows  that $k(\sigma) $ changes by $\pi /6$ between $\sigma_0$ and $\sigma'$. Hence, $k'(\sigma)=O(1/\epsilon)$ on $\Delta$ or on some part $\Delta '$ of $\Delta$, i.e.,
 \begin{equation}\label{vg}
 V_g=\frac{L}{k'(\sigma)}=O(\epsilon),~~\sigma \in \Delta'.
 \end{equation}

  Note that matrix $T_{\sigma}$ has a pole if $n=0$.
Since matrices $T_{\sigma}$ and $M_{\sigma}$ differ by a rotation, one can easily choose $l_3$ such that  $F(\sigma)=\frac{1}{2}$Tr$M_{\sigma}$ has a pole at a point where $n=0$ (in fact, the pole of $F(\sigma)$ exists for all but specific values of $l_3$). Thus the following equation provides the poles of $F(\sigma)$ under an
appropriate choice of $l_3$ (see (\ref{csp})-(\ref{trl}))
\[
n=\langle \delta, (I-P^2)^{-1}S\delta)\rangle=\frac{1}{2}\langle \delta, [(I-P)^{-1}+(I+P)^{-1}]S\delta)\rangle=0,~~P=C+SB.
\]

Let us introduce the matrix
\begin{equation} \label{ttt}
T=\frac{I-C}{S}=\left(\begin{array}{c}
x~~0\\0~~y \end{array}\right),~~~x=\tan \frac{\sigma l_1}{2},~~y=\tan \frac{\sigma l_2}{2}.
\end{equation}
Then the equation $n=0$ can be rewritten in the form
\begin{equation} \label{fff}
2n=\langle \delta, [(T-B)^{-1}+(T^{-1}+B)^{-1}]\delta)\rangle=0.
\end{equation}
We will conclude this subsection by an re-writing (\ref{fff}) using the components of matrix $B$ and vector $\delta$ (see (\ref{AB})):
\[
2n=\frac{(y-a_{2,2})\delta_1^2+(x-a_{1,1})\delta_2^2+
2a_{1,2}\delta_1\delta_2}{(x-a_{1,1})(y-a_{2,2})-a_{1,2}^2}
\]
\begin{equation} \label{f}
+\frac{(y^{-1}+a_{2,2})\delta_1^2+(x^{-1}+a_{1,1})\delta_2^2-
2a_{1,2}\delta_1\delta_2}{(x^{-1}+a_{1,1})(y^{-1}+a_{2,2})-a_{1,2}^2}=0.
\end{equation}
This is an algebraic equation of the forth order with respect to $(x,y)$. We will discuss it in more detail later.

 \textit{Transparency, preliminary discussion}. Let us recall the estimate (see \cite{smv2}) for the reflection coefficient $r=r_N$ by a finite slab of  periodic media which consists of $N$ periods:
\begin{equation} \label{pole}
|r_N|=|\frac{\sin Nk(\sigma)}{\sin k(\sigma)}|(||M_{\sigma}||^2-2)^{1/2},
\end{equation}
where $||M_{\sigma}||$ is the Gilbert-Shmidt norm of the monodromy matrix, i.e.
\[
||M_{\sigma}||^2=||(m_{i,j})||^2=\sum_{i,j \leq 2} m_{i,j}^2.
\]
We will choose parameters $l_j$ in such a way that the frequency support $\Delta$ of the wave packet is in the middle of a band where $|F(\sigma)|<1$. Then $|\sin k(\sigma)|>\frac{\sqrt{3}}{2}$ there, and
\begin{equation} \label{pole2}
|r_N(\sigma)|=\frac{2}{\sqrt{3}}|(||M_{\sigma}||^2-2)^{1/2},~~\sigma \in \Delta.
\end{equation}

We will choose $l_1,~l_2$ in such a way that $||M_{\sigma_0}||^2=2$ , and therefore
\begin{equation} \label{mo}
r_N(\sigma_0)=0.
\end{equation}
Then $r_N(\sigma)$ is small in a small neighborhood of $\sigma_0$, and we will have both the slowing down of the light and the transparency in $\Delta'$.

Note that $||M_{\sigma}||=||T_{\sigma}||$ (see (\ref{mon})) and that det$T_{\sigma}=1$. Thus $||T_{\sigma}||\geq 2$ and $||T_{\sigma}||= 2$ if matrix $T_{\sigma}$ is
orthogonal (the sum of its off-diagonal elements is zero). Hence, condition
\begin{equation} \label{mo1}
(m+n)(m-n)+1=m^2-n^2+1=0
\end{equation}
provides (\ref{mo}). We are going to write this condition more explicitly.

From (\ref{mn}), (\ref{trl}) it follows that
\[
m+n=c+\langle\delta, (I-P)^{-1}S \delta \rangle =c+\langle\delta, [S^{-1}(I-C-SB)]^{-1} \delta \rangle=
c+\langle\delta, [(T-B)]^{-1} \delta \rangle,
\]
where $T$ is defined in (\ref{ttt}).
Similarly
 \[
 m-n=c-\langle\delta, (I+P)^{-1}S \delta \rangle =c-\langle\delta, [S^{-1}(I+C+SB)]^{-1} \delta \rangle=
c-\langle\delta, [(T^{-1}+B)]^{-1} \delta \rangle.
 \]
 We substitute the last two formulas into (\ref{mo1}) and obtain the following exact form for (\ref{mo}):
\[
[c+\frac{(y-a_{2,2})\delta_1^2+(x-a_{1,1})\delta_2^2+
2a_{1,2}\delta_1\delta_2}{(x-a_{1,1})(y-a_{2,2})-a_{1,2}^2}]\cdot
\]
\begin{equation} \label{rnrn}
\cdot[c-\frac{(y^{-1}+a_{2,2})\delta_1^2+(x^{-1}+a_{1,1})\delta_2^2-
2a_{1,2}\delta_1\delta_2}{(x^{-1}+a_{1,1})(y^{-1}+a_{2,2})-a_{1,2}^2}]+1=0.
\end{equation}
Thus, the set of transparency points $\{r_N=0\}$ is also given by zeroes of a polynomial of forth order in
$(x,y)$-plane, $x=\tan \frac{\sigma l_1}{2},~y=\tan \frac{\sigma l_2}{2}$. To be more accurate, we need to omit points where $n=0$ from this set, since (\ref{mo1}) provides the orthogonality of $T_{\sigma}$ only if $n\neq0$.

 \textit{The choice of parameters.} The center $\sigma_0$ in the frequency interval of the wave packet is given. We need to choose $l_1,~l_2$ in such a way that (\ref{mo})-(\ref{rnrn}) hold at $\sigma_0$ and this point is close to a point where (\ref{f}), (\ref{fff}) hold. Note that condition (\ref{mo}) is equivalent to the orthogonality of matrix $T_{\sigma_0}$. Then one can easily transfer it by rotation to a matrix with zero trace, i.e., one can find $l_3$ such that (\ref{f0}) holds. Hence, it remains to choose $l_1,~l_2$ appropriately.

 We note that equation (\ref{mo1}) for $\sigma_0$ contradicts the condition $n=0$ imposed by  (\ref{f}), (\ref{fff}). Moreover, $n$ can not vanish at a point close to $\sigma_0$ if $n$ is smooth. In order for equation (\ref{mo1}) to be valid at a point $\sigma_0$ and $n$ to be zero at $\sigma_1$ which is close to $\sigma_0$, function $n$ has to be singular near $\sigma_0$. Thus one must choose parameters $(x,y)$ near the point where one of the denominators in (\ref{f}) is zero. Let us choose the denominator of the first fraction. Then the numerator of the first fraction also must be small (otherwise (\ref{f}) is not valid at $\sigma=\sigma_1$). By equating both the numerator and the denominator of the first fraction in (\ref{f}) to zero we find that
\begin{equation*}
x\approx a_{1,1}-a_{1,2}\frac{\delta_1}{\delta_2}, ~~y \approx a_{2,2}-a_{1,2}\frac{\delta_2}{\delta_1}.
\end{equation*}

Finally, we fix small $\epsilon$ and choose
\begin{equation} \label{xx}
x= a_{1,1}-a_{1,2}\frac{\delta_1}{\delta_2}+\epsilon.
\end{equation}
Then we solve (\ref{rnrn}) for $y$ in a neighborhood of the point $a_{2,2}-a_{1,2}\frac{\delta_2}{\delta_1}$. This implies
\begin{equation} \label{yy}
y=a_{2,2}-a_{1,2}\frac{\delta_2}{\delta_1}-(\frac{\delta_2}{\delta_1})^2\epsilon+\gamma\epsilon^2+
\text{O}(\epsilon)^3,
\end{equation}
where the exact value of $\gamma$ can be easily found from (\ref{rnrn}). We put $\sigma=\sigma_0$ in (\ref{ttt}) and determine   $l_1,~l_2$ from (\ref{xx}), (\ref{yy}). This choice of   $l_1,~l_2$ implies
(\ref{mo}).

By solving  (\ref{f}) asymptotically, we find that its solution also has the form (\ref{xx}), (\ref{yy}) with a different value of $\gamma$. This justifies the existence of a pole of the Hill discriminant at the distance O$(\epsilon^2)$ from the point $\sigma_0.$ In fact, let us omit the cubic term in (\ref{yy}). Then equations  (\ref{xx}), (\ref{yy}) with two different values of $\gamma$ define two parabolas $P_1$ and $P_2$ with the same vertex and tangent line at the vertex. If we put  $\sigma=\sigma_0$ in (\ref{ttt}) and fix $\epsilon$, we get a point on $P_1$ at the distance of order $\epsilon$ from the vertex which defines $l_1,l_2$. This choice implies (\ref{mo}). When $\sigma$ changes: $\sigma=\sigma_0+\tau,~0\leq \tau \leq \tau_0$, point $(x,y)$ moves in the direction of the vector $(l_1,l_2)$ or in the opposite direction (see \ref{ttt}). We may change the sign of $\tau$, if needed, to guarantee that the point moves toward the second parabola where $n=0$. The only unacceptable situation is when the vector $(l_1,l_2)$ is tangent to $P_1$. It will not happen in a generic case. Besides one can always avoid it by changing  $l_1$ or $l_2$, since they are defined up to an integer multiple of $2\pi$.

The arguments above prove that wave packet with frequencies in O$(\epsilon^2)$-semi-neighborhood of  $\sigma_0$ will propagate with the group velocity $V_g=$O$(\epsilon^2)$ and the reflection coefficient will have order O$(\epsilon^2)$.

We conclude this subsection by the following remark. Let the geometry of the junctions be chosen. After that, we must choose specific  $l_1,l_2$ satisfying (\ref{xx}), (\ref{yy}). It is easy to choose one of these parameters (for example, $l_1$) by changing the distance between two neighboring junctions, but $l_2$ will be defined after that by the geometry of the network. In fact it is not the geometrical, but only the optical length which plays role here. One can preserve the geometry of the network and change the refraction index in some of the channels to satisfy the relations (\ref{xx}), (\ref{yy}). This produces the same effect as changing the lengths of the corresponding channels.
\bigskip

\textbf{Acknowledgment.} The authors were supported partially by the NSF grant
DMS-0706928.


\begin{thebibliography}{9}


\bibitem{IT}Dell'Antonio, G., Tenuta, L.: \textit{Quantum graphs as
holonomic constraints,} J. Math. Phys., 47 (2006), 072102:1-21.

\bibitem{DE}Duclos P., Exner, P.: \textit{Curvature-induced bound states
in quantum waveguides in two and three dimensions,} Rev. Math. Phys., 7
(1995), 73-102.

\bibitem{DE3}Duclos P., Exner P., Stovicek P.: \textit{Curvature-induced
resonances in a two-dimensional Dirichlet tube,} Ann. Inst. H. Poincare 62
(1995), 81-101.

\bibitem{EP}Exner P., Post O.: \textit{Convergence of spectra of
graph-like thin manifolds,} J. Geom. Phys., 54 (2005), 77-115.

\bibitem{ES}Exner P., \v{S}eba P.: \textit{Electrons in semiconductor
microstructures: a challenge to operator theorists,} in Schr\"{o}dinger
Operators, Standard and Nonstandard (Dubna 1988), World Scientific,
Singapure (1989), 79-100.

\bibitem{ExS2}Exner P. and \v{S}eba P.: \textit{Bound states in curved
quantum waveguides,} J. Math. Phys. 30(1989), 2574 - 2580.

\bibitem{ExS4}Exner P., \v{S}eba P.: \textit{Trapping modes in a curved
electromagnetic waveguide with perfectly conducting walls, }Phys. Lett. A144
(1990), 347-350.

\bibitem{ExV2}Exner P. and Vugalter S. A.: \textit{Asymptotic estimates
for bound states in quantum waveguides coupled laterally through a narrow
window, }Ann. Inst. H. Poincare, Phys. Theor. 65 (1996), 109 - 123.

\bibitem{ExV3}Exner P. and Vugalter S. A.:  \textit{On the number of
particles that a curved quantum waveguide can bind}, J. Math. Phys. 40
(1999), 4630-4638.

\bibitem{ExW}Exner P., Weidl T.: \textit{Lieb-Thirring inequalities on
trapped modes in quantum wires}, Proceedings of the XIII International
Congress on Mathematical Physics (London 2000), International Press of
Boston, 2001, pp.437-443.

\bibitem{FW}Freidlin M., A. Wentzel A.: \textit{Diffusion processes on
graphs and averaging principle}, Ann. Probab., Vol 21, No 4 (1993),
2215-2245.

\bibitem{FW1}Freidlin M.: \textit{Markov Processes and Differential
Equations: Asymptotic Problems,} Lectures in Mathematics, ETH Zurich,
Birkhauser Verlag, Basel, 1996.

\bibitem{KS}Kostrykin V., Schrader R.: \textit{Kirchhoff's rule for
quantum waves}, J. Phys. A: Mathematical and General, Vol 32 (1999), 595-630.

\bibitem{K}Kuchment P.: \textit{Graph models of wave propagation in thin
structures}, Waves in Random Media, Vol.12 (2002), 1-24.

\bibitem{K1}Kuchment P.: \textit{Quantum graphs. I. Some basic structures}%
, Waves in Random Media, Vol.14, No 1 (2004), 107-128.

\bibitem{K2}Kuchment P.: \textit{Quantum graphs. II. Some spectral
properties of quantum and combinatorial graphs,} Journal of Physics A:
Mathematical and General, Vol 38, No 22 (2005), 4887-4900.

\bibitem{KZ1}Kuchment P., Zeng H.: \textit{Convergence of spectra of
mesoscopic systems collapsing onto a graph}, J. Math. Anal. Appl. 258
(2001), 671-700.

\bibitem{KZ2}Kuchment P., Zeng H.: \textit{Asymptotics of spectra of
Neumann Laplacians in thin domains}, in Advances in Differential Equations
and mathematical Physics, Yu. Karpeshina etc (Editors), Contemporary
Mathematics, AMS, 387 (2003), 199-213.

\bibitem{smv}  S. Molchanov, B. Vainberg, \ Slowdown of the wave packet in
finite slabs of periodic media, \textit{Waves in Random Media}, \textbf{14}
(2004), 411-423.

\bibitem{smv1}  S. Molchanov, B. Vainberg, \ Slowing down of wave packets in
quantum graphs, \textit{Waves in Complex and Random Media,} \textbf{15}, No 1 (2005), 101-112.

\bibitem{smv2}  S. Molchanov, B. Vainberg, \ Slowing down and reflection of
waves in trancated periodic media, J. of Funct. Analysis, 231 (2006),
287-311.

\bibitem{MV}  Molchanov S. and Vainberg B.: \textit{Transition from a
network of thin fibers to quantum graph: an explicitly solvable model},
Contemporary Mathematics, 415, AMS (2006), 227-240 (arXiv:math-ph/0605037).

\bibitem{MV1}  Molchanov S. and Vainberg B.: \textit{Scattering solutions in
networks of thin fibers: small diameter asymptotics,} Comm. Math. Phys.,
273, N2, (2007), 533-559 (arXiv:math-ph/0609021).

\bibitem{MV2}  Molchanov S. and Vainberg B.: \textit{Laplace operator in
networks of thin fibers: spectrum near the threshold, }in\textit{\ }%
Stochastic analysis in mathematical physics, 69--93, World Sci. Publ.,
Hackensack, NJ, 2008 (arXiv:0704.2795).

\bibitem{MVspain} Molchanov S. and Vainberg B.: Propagation of Waves in Networks of Thin Fibers, submitted

\bibitem{Pavlov4}Mikhailova, A., Pavlov, B., Popov, I., Rudakova, T.,
Yafyasov, A.: \textit{Scattering on a compact domain with few semi-infinite
wires attached: resonance case. Math. Nachr}. 235 (2002), 101--128.

\bibitem{Pavlov2}Pavlov B., Robert K.: \textit{Resonance optical switch:
calculation of resonance eigenvalues. }Waves in periodic and random media
(South Hadley, MA, 2002), 141--169, Contemp. Math., 339, Amer. Math. Soc.,
Providence, RI, 2003.

\bibitem{P}Post O.: \textit{Branched quantum wave guides with Dirichlet
BC: the decoupling case}, Journal of Physics A: Mathematical and General,
Vol 38, No 22 (2005), 4917-4932.

\bibitem{P1}Post O.: \textit{Spectral convergence of non-compact
quasi-one-dimensional spaces,} Ann. Henri Poincar, 7 (2006), 933-973.

\bibitem{RS}Rubinstein J., Schatzman M.: \textit{Variational problems on
multiply connected thin strips. I. Basic estimates and convergence of the
Laplacian spectrum}, Arch. Ration. Mech. Anal., 160 (2001), No 4, 293-306.

\end{thebibliography}
\end{document}